\documentclass[letterpaper, 10pt, conference]{ieeeconf}
\IEEEoverridecommandlockouts
\usepackage{amsmath,amssymb,amsfonts,mathtools,bm}

\usepackage{amsthm}
\usepackage{graphicx}
\usepackage{booktabs}
\usepackage{multirow}
\usepackage{hyperref}
\usepackage{xcolor}
\usepackage{algorithm}
\usepackage{algorithmic}
\usepackage{cite}

\newif\iftrackchanges \trackchangesfalse
\newcommand{\added}[1]{\iftrackchanges{\color{blue}#1}\else#1\fi}
\newcommand{\deleted}[1]{\iftrackchanges{\color{red}[DEL: #1]}\else\fi}
\newcommand{\replaced}[2]{\iftrackchanges{\color{blue}#1}{\color{red}\,[DEL: #2]}\else#1\fi}

\hypersetup{
    colorlinks=true,
    linkcolor=blue,
    citecolor=blue,
    urlcolor=blue
}

\newtheorem{proposition}{Proposition}
\newtheorem{corollary}{Corollary}
\newtheorem{assumption}{Assumption}
\newtheorem{remark}{Remark}
\newtheorem{definition}{Definition}
\newtheorem{theorem}{Theorem}
\newtheorem{lemma}{Lemma}

\newcommand{\tr}{\operatorname{Tr}}
\newcommand{\vecop}{\operatorname{vec}}
\newcommand{\argmin}{\operatorname*{arg\,min}}
\newcommand{\R}{\mathbb{R}}
\newcommand{\IFm}{\mathrm{IF}^{\mathrm{m}}}
\newcommand{\IFfixed}{\mathrm{IF}^{\mathrm{fixed}}}
\newcommand{\IFstoch}{\mathrm{IF}^{\mathrm{stoch}}}

\title{\LARGE \bf
Trajectory Influence Functions for Stochastic LQR: Joint Sensitivity Through Dynamics and Noise Covariance
}

\author{
Jiachen Li$^{1}$, Shihao Li$^{1}$, Soovadeep Bakshi$^{1}$, Jiamin Xu$^{1}$, and Dongmei Chen$^{1}$%
\thanks{$^{1}$All authors are with the Department of Mechanical Engineering, The University of Texas at Austin, Austin, TX 78712, USA.
{\tt\small \{jiachenli, shihaoli01301, soovadeepbakshi, jiaminxu\}@utexas.edu, dmchen@me.utexas.edu}}%
}

\begin{document}
\maketitle
\thispagestyle{empty}
\pagestyle{empty}

\begin{abstract}
\replaced{We study trajectory-level data attribution for the plug-in stochastic LQR
cost $\hat J(\mathcal D)=\tr(P(\hat\theta)\hat W)$, in which both the dynamics
$(A,B)$ and the process-noise covariance are estimated from the same
identification data. Our target is the leave-one-trajectory-out (LOTO) shift
$\Delta\hat J_k$, and we develop a first-order influence approximation that
isolates the identification-to-control coupling through a three-level
hierarchy: (i) the model-side surrogate $\IFm_k:=H^{-1}g_k$, which is the
\emph{exact} ridge-LS parameter shift rather than a numerical Newton step;
(ii) the fixed-covariance score $\IFfixed_k=\zeta_\Sigma^\top\IFm_k$; and
(iii) the stochastic score $\IFstoch_k$, obtained from $\IFm_k$ by adding a
residual-channel gradient $h$ together with an exact direct-removal covariance
correction. Theorem~\ref{thm:decomp} decomposes $\Delta\hat J_k$ into a
leading first-order term and three explicit remainders. The new covariance
remainder admits an algebraic bound and is free of Lyapunov amplification.
Experiments on a DC motor, a mass--spring--damper with heterogeneous noise,
and a quadrotor UAV hover task show near-perfect agreement with exact LOTO
retraining on linear plants and on the linearized hover regime, so the score
serves as a practical substitute for retraining in linear stochastic LQR.}{We present a three-level influence hierarchy for data valuation in stochastic
LQR. At the \emph{model level}, the trajectory influence surrogate
$\IFm_k := H^{-1}g_k$ approximates the leave-one-trajectory parameter
shift. At the \emph{control level} with fixed covariance, the usual
fixed-noise score is obtained by composing $\IFm_k$ with the Riccati
gradient of $\tr(P(\theta)\Sigma)$. At the \emph{stochastic control level},
the plug-in cost depends additionally on the residual covariance estimate
$\hat W$, so removing a trajectory perturbs the cost through both the dynamics
and the covariance channels. We derive an exact leave-one-trajectory
decomposition of the covariance shift into a \emph{direct-removal} term and a
\emph{parameter-shift} term, show that the additional first-order contribution
is a simple residual cross-moment, and obtain a stochastic influence score
built directly on $\IFm_k$. The resulting method preserves the
amortized structure of prior work: after one Hessian factorization and one
adjoint Lyapunov solve, each trajectory requires only a dot product plus an
$O(n_x^2)$ direct-removal correction. The shared Hessian solves can also be
performed iteratively by conjugate gradients when explicit factorization is
undesirable. The new covariance remainder is explicit
and does not involve Lyapunov-operator amplification; the only amplified term is
the familiar Riccati remainder inherited from fixed-covariance influence
analysis. Numerical results on two linear systems show that accounting for the
estimated covariance substantially improves agreement with exact
leave-one-trajectory retraining, especially under heterogeneous noise.}
\end{abstract}

\begin{figure*}[!t]
\centering
\includegraphics[width=0.85\textwidth]{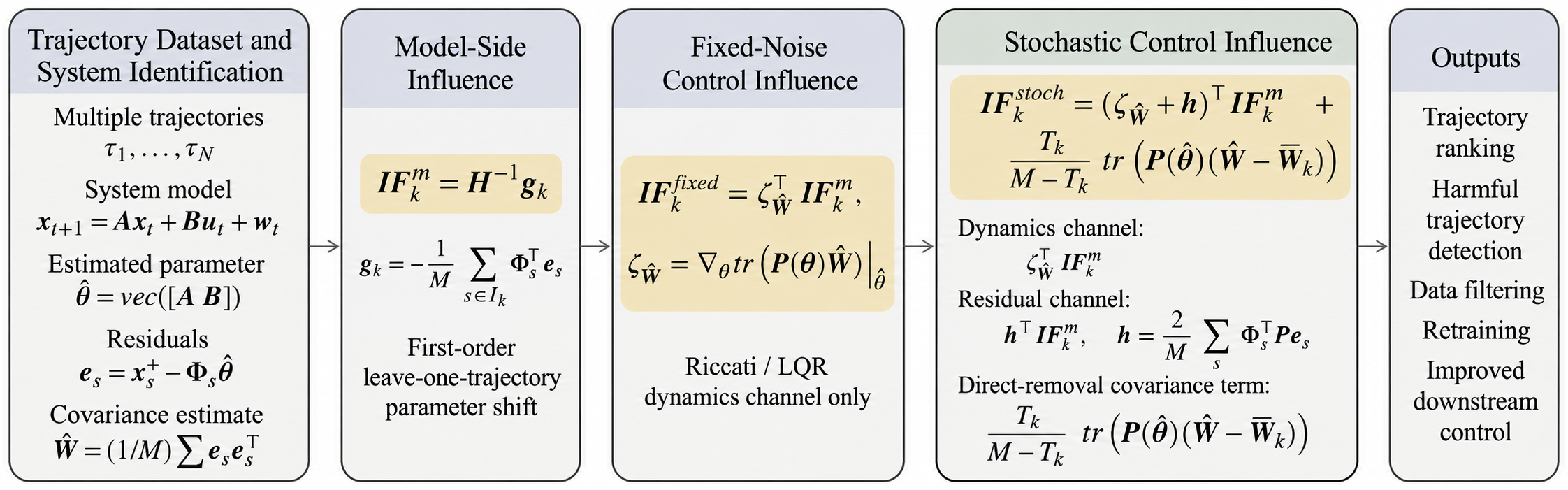}
\caption{Stochastic-LQR trajectory-influence pipeline for
$\Delta\hat J_k=\hat J(\mathcal D)-\hat J(\mathcal D\setminus k)$ with
$\hat J(\mathcal D)=\tr(P(\hat\theta)\hat W)$. Stages~1--3 (gray) follow~\cite{li2025influence};
Stage~4 (yellow, NEW) adds residual gradient $h$ and the direct-removal term.}
\label{fig:overview}
\end{figure*}

\section{Introduction}
\label{sec:intro}

Data valuation for control, by which we mean a way of quantifying which
training trajectories matter most for downstream closed-loop performance, is
becoming consequential as model-based control moves into data-limited and
safety-critical settings~\cite{dean2020sample,tu2019gap}. Such
attributions support active data collection, dataset auditing under
distribution shift, and the diagnosis of outliers whose effect on the
learned controller is disproportionate to their abundance in the
identification set. A natural
notion of value is the \emph{leave-one-trajectory-out} (LOTO) change in a
plug-in performance metric: by how much does the estimated LQR cost shift
when trajectory $k$ is removed from the identification set?

Computing $\Delta\hat J_k$ exactly requires refitting the model $N$ times,
which is prohibitive for large datasets and incompatible with online or
streaming settings. Influence functions avoid this cost by replacing each
refit with a first-order approximation of the parameter
shift~\cite{koh2017understanding,basu2021influence}, and beyond the
computational savings they expose the structural origin of each trajectory's
effect (dynamics-side Riccati, residual covariance, finite-sample direct
removal) that the scalar $\Delta\hat J_k$ alone hides.
For LQR specifically,
\cite{li2025influence} showed that with the noise covariance held fixed, the
sensitivity of the discrete algebraic Riccati equation (DARE) cost can be
computed efficiently through an adjoint Lyapunov equation.

\added{
On linear plants this yields a practical substitute for retraining at $O(p)$
per trajectory after one factorization; the nonlinear extension is left for
future work.}

The covariance, however, is rarely known in advance. In practice the
plug-in cost uses an estimate built from the same residuals that determine
$(A,B)$. Removing trajectory $k$ therefore perturbs the cost through two
coupled mechanisms: it shifts the dynamics estimate and it shifts the
residual covariance. Freezing the covariance across LOTO splits suppresses
both effects, ignoring the finite-sample direct-removal contribution as well
as the first-order change induced by the parameter perturbation.

\added{The quantity of interest is the LOTO shift of the plug-in cost,
\begin{equation}
  \Delta\hat J_k
  := \tr\bigl(P(\hat\theta_{\setminus k})\hat W_{\setminus k}\bigr) - \tr(P_0\hat W),
  \label{eq:target_intro}
\end{equation}
for which Theorem~\ref{thm:decomp} establishes the decomposition
\begin{multline}
  \Delta\hat J_k
  = \bigl(\zeta_{\hat W}+h\bigr)^\top\delta\theta_k
  + \tfrac{T_k}{M_{\setminus k}}\tr\bigl(P_0(\hat W-\bar W_k)\bigr) \\
  + \mathcal R_k^{\mathrm{Ric}} + \mathcal R_k^{W} + \mathcal R_k^{\times},
  \label{eq:decomp_intro}
\end{multline}
combining a first-order coupling term, an exact direct-removal correction, and
three higher-order remainders.}

\added{Building on the fixed-covariance LQR influence framework
of~\cite{li2025influence}, which provides the model-side object
$\IFm_k=H^{-1}g_k$, the Riccati gradient $\zeta_\Sigma$, and the score
$\IFfixed_k(\Sigma)=\zeta_\Sigma^\top\IFm_k$, this letter contributes the
residual-channel gradient $h:=(2/M)\sum_s\Phi_s^\top P_0 e_s$ that captures
$\partial\hat W/\partial\theta$; the exact direct-removal correction
$(T_k/M_{\setminus k})\tr\bigl(P_0(\hat W-\bar W_k)\bigr)$; an algebraic bound
on the covariance remainder $\mathcal R_k^W$ that avoids Lyapunov
amplification (Proposition~\ref{prop:param_shift}); and the stochastic score
$\IFstoch_k=(\zeta_{\hat W}+h)^\top\IFm_k+(T_k/M_{\setminus k})\tr\bigl(P_0(\hat W-\bar W_k)\bigr)$,
which reduces to $\IFfixed_k(W)$ when $\Sigma=W$ is known a priori
(Corollary~\ref{cor:reduction}).}

\textbf{Related work.}
Influence functions trace back to robust statistics
\cite{hampel1974influence,ling1984residuals} and were adapted to modern
machine learning by~\cite{koh2017understanding}, with related point-level
attribution via representer points~\cite{yeh2018representer} and
Shapley values~\cite{ghorbani2019data}, trajectory/sequence
variants in~\cite{basu2021influence}, and a
first-order-fidelity analysis in~\cite{bae2022if}. On the control side,
\cite{dean2020sample,mania2019certainty,tu2019gap} study
LQR from data; \cite{li2025influence} introduced an efficient
trajectory-level control influence score for LQR with fixed noise
covariance. At scale, attribution is accelerated by Hessian-vector-product
solvers or by checkpoint surrogates such as TracIn and
TRAK~\cite{pruthi2020estimating,pmlr-v202-park23c}. Our contribution extends the
fixed-covariance LQR framework to the practically relevant plug-in cost
$\tr(P(\hat\theta)\hat W)$, in which both the dynamics and the covariance are
estimated from data, and exposes the resulting stochastic score as a natural
successor of a more primitive model-side influence object.

\section{Preliminaries}
\label{sec:prelim}

\subsection{Identification model}

Consider the discrete-time linear system
\begin{equation}
    x_{t+1} = A x_t + B u_t + w_t,
    \qquad w_t \sim \mathcal N(0,W), \quad W \succeq 0,
    \label{eq:sys}
\end{equation}
with state dimension $n_x$ and input dimension $n_u$. We observe $N$
trajectories $\{\tau_k\}_{k=1}^N$; trajectory $k$ contains $T_k$ transitions,
and $M := \sum_{k=1}^N T_k$ denotes the total number of transitions.

\added{
The regressor at transition $s$ stacks the state and input,
\begin{equation}
  \Phi_s := \begin{bmatrix} x_s^\top & u_s^\top \end{bmatrix} \otimes I_{n_x}
  \in \R^{n_x\times p}, \qquad p := n_x(n_x+n_u),
  \label{eq:phi_def}
\end{equation}
so that $\theta:=\vecop([A\;B])\in\R^p$ recovers $(A,B)$ via a fixed
permutation. The disturbances $\{w_s\}_{s=1}^M$ are assumed independent across
$s$ with $\mathbb E[w_s]=0$ and $\mathbb E[w_sw_s^\top]\preceq W$, and are
\emph{not} assumed identical across trajectories in the heterogeneous-noise
experiment of Section~\ref{sec:experiments}.}

For each transition $s$, let $x_s^+$ be the successor state and
$\Phi_s \in \R^{n_x \times p}$ the regressor matrix such that
\begin{equation}
    x_s^+ = \Phi_s \theta^* + w_s,
    \qquad \theta^* := \vecop([A\;\;B]) \in \R^p.
\end{equation}
The regularized least-squares estimator is
\begin{equation}
    \hat\theta
    = \argmin_{\theta \in \R^p}
    \Bigl\{
    \frac{1}{2M}\sum_{s=1}^M \|x_s^+ - \Phi_s\theta\|^2
    + \frac{\lambda}{2}\|\theta\|^2
    \Bigr\}.
    \label{eq:ls_estimator}
\end{equation}
\added{
For the ridge-LS objective in \eqref{eq:ls_estimator} with $\lambda>0$, the
full-data Hessian
$H=\tfrac{1}{M}\sum_s\Phi_s^\top\Phi_s+\lambda I_p \succeq \lambda I_p$
is invertible regardless of the rank of the regressor sample covariance.
We assume $\lambda>0$ throughout, so $H^{-1}$ in subsequent formulas is
well defined.}
Define the residuals
\begin{equation}
    e_s(\theta) := x_s^+ - \Phi_s\theta,
    \qquad e_s := e_s(\hat\theta),
\end{equation}
and the Hessian
\begin{equation}
    H := \frac{1}{M}\sum_{s=1}^M \Phi_s^\top \Phi_s + \lambda I_p.
\end{equation}
We use the trajectory-level gradient contribution
\begin{equation}
    g_k := \nabla_\theta\!\biggl[
    \frac{1}{2M}\sum_{s\in\mathcal I_k}
    \|x_s^+ - \Phi_s\theta\|^2
    \biggr]_{\theta = \hat\theta}
    = -\frac{1}{M}\sum_{s\in\mathcal I_k} \Phi_s^\top e_s,
    \label{eq:gk_def}
\end{equation}
where $\mathcal I_k$ is the transition index set of trajectory $k$.
With this sign convention, the standard group-influence surrogate for the LOTO
parameter shift is $\widetilde\delta\theta_k := H^{-1} g_k$.

\replaced{
\begin{remark}[About the parameter-shift surrogate]
\label{rem:param_surrogate}
The exact LOTO estimator $\hat\theta_{\setminus k}$ minimizes the ridge-LS
objective renormalized by $M_{\setminus k}:=M-T_k$, with first-order
optimality $H_{\setminus k}\,\delta\theta_k =
g_k - (T_k/M_{\setminus k})\lambda\hat\theta + r_k$. The surrogate
$\widetilde\delta\theta_k = H^{-1}g_k$ replaces $H_{\setminus k}^{-1}$ by
$H^{-1}$ and drops $(T_k/M_{\setminus k})\lambda\hat\theta$. Because the loss
is quadratic, $\widetilde\delta\theta_k$ is the \emph{exact closed-form
parameter shift of the full-data quadratic loss}---not a numerical Newton
step---but still a surrogate for
$\delta\theta_k:=\hat\theta_{\setminus k}-\hat\theta$. Its approximation error
is kept explicit (Proposition~\ref{prop:error_modular}) and quantified in
Section~\ref{sec:experiments}.
\end{remark}
}{\begin{remark}[About the parameter-shift surrogate]
\label{rem:param_surrogate}
The exact LOTO estimator $\hat\theta_{\setminus k}$ minimizes the objective
renormalized by $M_{\setminus k}:=M-T_k$, whereas the standard influence
surrogate replaces this with the first-order perturbation induced by removing the
trajectory loss from the full objective. Consequently,
$\widetilde\delta\theta_k = H^{-1}g_k$ is a \emph{surrogate} for the exact
shift $\delta\theta_k := \hat\theta_{\setminus k} - \hat\theta$, and the final
score inherits whatever approximation error is incurred in this step. We keep
that error explicit below.
\end{remark}}

\added{
\begin{remark}[Ridge renormalization correction]
\label{rem:ridge_renorm}
The term $(T_k/M_{\setminus k})\lambda\hat\theta$ that appears in the LOO
optimality condition reflects the fact that the renormalized LOO objective
weights the ridge regularizer relative to the smaller dataset size
$M_{\setminus k}$. For $T_k \ll M$ this correction is $O(N^{-1})$ and is
absorbed into the surrogate error of Remark~\ref{rem:param_surrogate}.
\end{remark}}

\begin{definition}[Model-side influence (IF$^{\mathrm{m}}$)]
\label{def:IF0}
For trajectory $k$, define the model-side influence surrogate
\begin{equation}
    \IFm_k := \widetilde\delta\theta_k = H^{-1}g_k \in \R^p.
    \label{eq:IF0}
\end{equation}
Thus $\IFm_k$ is the first-order approximation of the exact LOTO
parameter shift $\delta\theta_k = \hat\theta_{\setminus k}-\hat\theta$.
It is a \emph{vector-valued} influence object: its role is to summarize how the
identified model changes when trajectory $k$ is removed.
\end{definition}

\begin{remark}[Why IF$^{\mathrm{m}}$ is the right foundation]
\label{rem:if0_foundation}
A scalar model-fit score can be obtained by pairing $\IFm_k$ with the
gradient of a chosen identification metric, but for downstream control analysis
the vector form is the useful primitive. Once $\IFm_k$ is available,
any first-order trajectory influence score is obtained by applying the gradient
of the target performance functional to $\IFm_k$.
\end{remark}

\subsection{Riccati notation}

Throughout, $(A,B)$ denotes the matrix pair encoded by $\hat\theta$, unless
explicit dependence on $\theta$ is written explicitly. Let $P_0 := P(\hat\theta)$
be the stabilizing DARE solution,
\begin{equation}
    P_0 = Q + A^\top P_0 A
    - A^\top P_0 B(R + B^\top P_0 B)^{-1} B^\top P_0 A,
    \label{eq:dare}
\end{equation}
\added{
where $Q\in\R^{n_x\times n_x}$ is the symmetric positive semi-definite state
weight, $R\in\R^{n_u\times n_u}$ is the symmetric positive definite input
weight, and \eqref{eq:dare} is the discrete algebraic Riccati equation (DARE).
Under Assumption~\ref{ass:stab} the stabilizing solution $P(\theta)$ exists,
is unique, is positive semi-definite, and depends twice continuously
differentiably on $\theta$ in a neighborhood of $\hat\theta$.}
with associated feedback gain
\begin{equation}
    K_0 := (R + B^\top P_0 B)^{-1}B^\top P_0 A,
\end{equation}
and closed-loop matrix $A_{\mathrm{cl}} := A - BK_0$.

\begin{assumption}[Local Riccati regularity]
\label{ass:stab}
At $\hat\theta$, the pair $(A,B)$ is stabilizable and $(A,Q^{1/2})$ is
detectable. \added{Furthermore, $Q\succeq 0$ and $R\succ 0$.}%
{} Moreover, there exists a neighborhood $\mathcal U$ of $\hat\theta$
on which the stabilizing DARE solution $P(\theta)$ exists uniquely and is twice
continuously differentiable. \added{This neighborhood condition is required for
each $\hat\theta_{\setminus k}$ used in Theorem~\ref{thm:decomp}; we assume
$\hat\theta_{\setminus k}\in\mathcal U$ for all $k$ throughout.}
\end{assumption}

Assumption~\ref{ass:stab} guarantees that $A_{\mathrm{cl}}$ is Schur stable and
that first-order Riccati sensitivities are well defined in a neighborhood of the
fitted model.

\begin{lemma}[Fixed-covariance Riccati gradient, cf.~\cite{li2025influence}]
\label{lem:IF2}
Let $\Sigma \succeq 0$ be fixed, and let $\Lambda_\Sigma$ solve the discrete
Lyapunov equation
\begin{equation}
    \Lambda_\Sigma - A_{\mathrm{cl}}\Lambda_\Sigma A_{\mathrm{cl}}^\top = \Sigma.
\end{equation}
\added{
Here $\Sigma\succeq 0$ plays the role of an arbitrary noise covariance at which
the Riccati cost $\tr(P(\theta)\Sigma)$ is to be differentiated, and
$\Lambda_\Sigma\succeq 0$ is the closed-loop dual covariance through which the
sensitivity is computed. Concretely, for any direction $v\in\R^p$ encoding
$([\delta A\;\delta B])$,
\begin{equation}
  \nabla_\theta\tr(P(\theta)\Sigma)\big|_{\hat\theta}^{\!\top}\, v
  = 2\tr\bigl(\Lambda_\Sigma A_{\mathrm{cl}}^\top P_0 [\delta A - \delta B K_0]\bigr),
  \label{eq:zeta_explicit}
\end{equation}
which we use as a compact restatement of the formula derived in
\cite{li2025influence}.}
Then the gradient
$\zeta_\Sigma := \nabla_\theta\tr(P(\theta)\Sigma)\vert_{\hat\theta}$ can be
computed from the adjoint Lyapunov/DARE formula in~\cite{li2025influence}.
\end{lemma}

\begin{definition}[Fixed-covariance control influence (IF$^{\mathrm{fixed}}$)]
\label{def:IF2}
For any fixed covariance matrix $\Sigma\succeq 0$, define
\begin{equation}
    \IFfixed_k(\Sigma) := \zeta_\Sigma^\top \IFm_k
    = g_k^\top H^{-1}\zeta_\Sigma.
    \label{eq:IF2_def}
\end{equation}
\added{
Operationally, $\IFfixed_k(\Sigma)$ is the first-order change in the
\emph{fixed-noise} cost $\tr(P(\theta)\Sigma)$ that would be incurred if the
identified parameter $\hat\theta$ were perturbed by the LOTO surrogate
$\IFm_k=H^{-1}g_k$, treating $\Sigma$ as held constant across the
LOTO split.}
In particular, the fixed-noise score studied in~\cite{li2025influence} is
recovered by taking $\Sigma=W$ when the process covariance is known, or
$\Sigma=\hat W$ when one freezes the plug-in covariance estimate.
\end{definition}

\begin{remark}[Hierarchy from model influence to control influence]
\label{rem:hierarchy}
Equation~\eqref{eq:IF2_def} makes the modeling-control separation explicit.
The object $\IFm_k$ is purely an identification-side quantity, while
$\IFfixed_k(\Sigma)$ is obtained by composing it with a control gradient.
The stochastic score derived later follows the same pattern but with two extra
covariance contributions: a residual-channel gradient and an exact
direct-removal term.
\end{remark}

\section{Plug-in stochastic cost}
\label{sec:stochastic_cost}

\begin{remark}[Plug-in target]
\label{rem:plugin}
\replaced{
The \emph{environment cost} is the infinite-horizon average stage cost
$J_{\mathrm{env}}:=\tr(P(\theta^*)W)$ that would be incurred under the true
parameters $(\theta^*,W)$. The \emph{plug-in cost} is its data-driven
proxy $\hat J(\mathcal D):=\tr(P(\hat\theta)\hat W)$ formed from the fitted
dynamics $\hat\theta$ and the residual covariance $\hat W$.}{The environment cost depends on the unknown triple $(A,B,W)$. In contrast, a
practitioner only has access to the plug-in estimate
$\hat J(\mathcal D):=\tr(P(\hat\theta)\hat W)$ built from the fitted dynamics and
a residual covariance estimate $\hat W$.}
Our influence score targets the LOTO change in this \emph{plug-in} quantity\added{; specifically, $\Delta\hat J_k$ in \eqref{eq:target_intro}}.
\end{remark}

\begin{proposition}[Average-cost identity]
\label{prop:telescope}
Under Assumption~\ref{ass:stab}, the infinite-horizon average stage cost of the
stabilizing LQR law associated with $P(\theta)$ and process covariance $W$ is
\begin{equation}
    J_{\mathrm{avg}}(\theta;W) = \tr(P(\theta)W).
\end{equation}
\end{proposition}

\begin{proof}
Let $K(\theta)$ be the stabilizing LQR gain and
$A_{\mathrm{cl}}(\theta)=A-BK(\theta)$. The \added{(unique stabilizing)} stationary state covariance
$\Sigma_{\mathrm{ss}}$ solves
$\Sigma_{\mathrm{ss}} = A_{\mathrm{cl}}\Sigma_{\mathrm{ss}}A_{\mathrm{cl}}^\top + W$.
The DARE identity
$Q + K^\top R K = P - A_{\mathrm{cl}}^\top P A_{\mathrm{cl}}$
therefore yields
\begin{align*}
    J_{\mathrm{avg}}
    &= \tr\bigl((Q + K^\top R K)\Sigma_{\mathrm{ss}}\bigr) \\
    &= \tr\bigl((P - A_{\mathrm{cl}}^\top P A_{\mathrm{cl}})\Sigma_{\mathrm{ss}}\bigr) \\
    &= \tr\bigl(P(\Sigma_{\mathrm{ss}} - A_{\mathrm{cl}}\Sigma_{\mathrm{ss}}A_{\mathrm{cl}}^\top)\bigr)
     = \tr(PW),
\end{align*}
where we used cyclicity of the trace in the third step.
\end{proof}

The residual covariance estimator is
\begin{equation}
    \hat W := \frac{1}{M}\sum_{s=1}^M e_s e_s^\top,
    \label{eq:What_def}
\end{equation}
so the plug-in cost is
$\hat J(\mathcal D)=\tr(P_0\hat W)$ with $P_0=P(\hat\theta)$.
The corresponding stochastic control influence score will be obtained by
starting from $\IFm_k$ and then adding the two covariance effects that
are absent from fixed-covariance control influence.

\section{LOTO shift of the residual covariance}
\label{sec:loto}

Let $\hat\theta_{\setminus k}$ be the LOTO estimator and define
\begin{equation}
    \delta\theta_k := \hat\theta_{\setminus k} - \hat\theta,
    \qquad M_{\setminus k} := M - T_k.
\end{equation}
The LOTO covariance estimate is
\begin{equation}
    \hat W_{\setminus k}
    := \frac{1}{M_{\setminus k}}\sum_{i\neq k}\sum_{s\in\mathcal I_i}
    e_s(\hat\theta_{\setminus k})e_s(\hat\theta_{\setminus k})^\top,
\end{equation}
and we write
$\Delta\hat W_k := \hat W_{\setminus k} - \hat W$.

\begin{assumption}[Bounded per-transition quantities]
\label{ass:bounded_data}
There exist finite constants $L_\Phi,L_e>0$ such that
$\|\Phi_s\|\le L_\Phi$ and $\|e_s\|\le L_e$ for all $s=1,\dots,M$.
\end{assumption}

\begin{assumption}[No dominant trajectory]
\label{ass:bounded_frac}
For every $k$, $T_k < M$. When asymptotic orders are stated, we assume in
addition that $\max_k T_k/M = O(N^{-1})$.
\end{assumption}

\begin{proposition}[Direct-removal term is exact]
\label{prop:direct}
Define the per-trajectory residual covariance
\begin{equation}
    \bar W_k := \frac{1}{T_k}\sum_{s\in\mathcal I_k} e_s e_s^\top.
\end{equation}
Then
\begin{equation}
    \Delta\hat W_k^{(\mathrm{dir})}
    := \frac{1}{M_{\setminus k}}\sum_{i\ne k}\sum_{s\in\mathcal I_i} e_s e_s^\top - \hat W
    = \frac{T_k}{M_{\setminus k}}(\hat W - \bar W_k).
    \label{eq:DW_direct}
\end{equation}
\end{proposition}

\begin{proof}
Using $M\hat W = \sum_{i=1}^N\sum_{s\in\mathcal I_i} e_s e_s^\top$ and
$T_k\bar W_k = \sum_{s\in\mathcal I_k} e_s e_s^\top$,
\begin{align*}
    \frac{1}{M_{\setminus k}}\sum_{i\ne k}\sum_{s\in\mathcal I_i} e_s e_s^\top
    &= \frac{1}{M_{\setminus k}}(M\hat W - T_k\bar W_k) \\
    &= \hat W + \frac{T_k}{M_{\setminus k}}(\hat W - \bar W_k).
\end{align*}
Subtracting $\hat W$ gives \eqref{eq:DW_direct}.
\end{proof}

\begin{proposition}[First-order expansion of the covariance shift]
\label{prop:param_shift}
Under Assumptions~\ref{ass:bounded_data}--\ref{ass:bounded_frac},
\begin{multline}
\Delta\hat W_k
= \tfrac{T_k}{M_{\setminus k}}(\hat W - \bar W_k) \\
- \tfrac{1}{M}\sum_{s=1}^M
\bigl(e_s\delta\theta_k^\top\Phi_s^\top + \Phi_s\delta\theta_k e_s^\top\bigr)
+ R_k^W,
\label{eq:DW_param}
\end{multline}
where the remainder \added{$R_k^W$ (we use a single notation throughout)} satisfies
\begin{equation}
    \|R_k^W\|_F
    \le L_\Phi^2\|\delta\theta_k\|^2
    + 4\frac{T_k}{M}L_eL_\Phi\|\delta\theta_k\|.
    \label{eq:DW_param_bound}
\end{equation}
In particular, if $T_k/M=O(N^{-1})$ and $\|\delta\theta_k\|=O(N^{-1})$, then
$\|R_k^W\|_F = O(N^{-2})$.
\end{proposition}

\begin{proof}
For each retained transition,
$e_s(\hat\theta_{\setminus k}) = e_s - \Phi_s\delta\theta_k$, so
$e_s(\hat\theta_{\setminus k})e_s(\hat\theta_{\setminus k})^\top
 = e_s e_s^\top - (e_s\delta\theta_k^\top\Phi_s^\top
 + \Phi_s\delta\theta_k e_s^\top) + \Phi_s\delta\theta_k\delta\theta_k^\top\Phi_s^\top$.
Averaging over $i\neq k$ and subtracting $\hat W$ isolates the direct term
$\Delta\hat W_k^{(\mathrm{dir})} :=
\tfrac{1}{M_{\setminus k}}\sum_{i\ne k}\sum_{s\in\mathcal I_i} e_s e_s^\top - \hat W$
plus cross- and quadratic-in-$\delta\theta_k$ sums.
Proposition~\ref{prop:direct} gives the direct term. We then add and subtract
$M^{-1}\sum_{s=1}^M(e_s\delta\theta_k^\top\Phi_s^\top + \Phi_s\delta\theta_k e_s^\top)$.
The resulting remainder is the sum of
\begin{align*}
    R_{k,1}^W &:= \frac{1}{M_{\setminus k}}\sum_{i\ne k}\sum_{s\in\mathcal I_i}
    \Phi_s\delta\theta_k\delta\theta_k^\top\Phi_s^\top, \\
    R_{k,2}^W &:= \frac{1}{M}\sum_{s\in\mathcal I_k}
    \bigl(e_s\delta\theta_k^\top\Phi_s^\top + \Phi_s\delta\theta_k e_s^\top\bigr), \\
    R_{k,3}^W &:= -\frac{T_k}{MM_{\setminus k}}
    \sum_{i\ne k}\sum_{s\in\mathcal I_i}
    \bigl(e_s\delta\theta_k^\top\Phi_s^\top + \Phi_s\delta\theta_k e_s^\top\bigr).
\end{align*}
Using
$\|uv^\top + vu^\top\|_F \le 2\|u\|\,\|v\|$,
$\|\Phi_s\delta\theta_k\delta\theta_k^\top\Phi_s^\top\|_F
 \le \|\Phi_s\|^2\|\delta\theta_k\|^2$,
and Assumption~\ref{ass:bounded_data}, we obtain
\begin{align*}
    \|R_{k,1}^W\|_F &\le L_\Phi^2\|\delta\theta_k\|^2, \\
    \|R_{k,2}^W\|_F &\le 2\frac{T_k}{M}L_eL_\Phi\|\delta\theta_k\|, \\
    \|R_{k,3}^W\|_F &\le 2\frac{T_k}{M}L_eL_\Phi\|\delta\theta_k\|.
\end{align*}
Summing the three bounds gives \eqref{eq:DW_param_bound}.
\end{proof}

\added{
\begin{remark}[No Lyapunov amplification in $R_k^W$]
\label{rem:no_lyap}
The bound \eqref{eq:DW_param_bound} is algebraic in $(\Phi_s,e_s,\delta\theta_k)$
and contains no $\|\Lambda_{\hat W}\|$ or closed-loop spectral factor; the
Riccati remainder $\mathcal R_k^{\mathrm{Ric}}$ in Theorem~\ref{thm:decomp}
does inherit Lyapunov amplification. We do not, however, claim a tighter ratio,
since $\mathcal R_k^{\mathrm{Ric}}$ remains dominant when the closed loop has
slow modes.
\end{remark}}

\begin{proposition}[Residual-channel gradient]
\label{prop:res_grad}
Define
\begin{equation}
    h := \frac{2}{M}\sum_{s=1}^M \Phi_s^\top P_0 e_s \in \R^p.
    \label{eq:h_def}
\end{equation}
Then the map $\theta \mapsto \tr(P_0\hat W(\theta))$ satisfies
\begin{equation}
    \nabla_\theta\tr(P_0\hat W(\theta))\big|_{\hat\theta} = -h.
\end{equation}
\added{
With $g_k=-(1/M)\sum_{s\in\mathcal I_k}\Phi_s^\top e_s$ in
\eqref{eq:gk_def} and $\widetilde\delta\theta_k=H^{-1}g_k$, the $+h$ in
$(\zeta_{\hat W}+h)^\top\IFm_k$ arises because $\Delta\hat J_k$ uses
$\hat W_{\setminus k}-\hat W$, which contributes $-\delta\theta_k$ to the
first-order term of \eqref{eq:DW_param} and flips the sign of $-h$;
\eqref{eq:DW_param_cost} below records the result.}
Consequently,
\begin{equation}
    \tr(P_0\Delta\hat W_k)
    = h^\top\delta\theta_k
    + \frac{T_k}{M_{\setminus k}}\tr\bigl(P_0(\hat W - \bar W_k)\bigr)
    + r_k^W,
    \label{eq:DW_param_cost}
\end{equation}
with
\begin{equation}
    |r_k^W| \le \|P_0\|\,\|R_k^W\|_F.
\end{equation}
\end{proposition}

\begin{proof}
Since
$\tr(P_0\hat W(\theta)) = M^{-1}\sum_{s=1}^M e_s(\theta)^\top P_0 e_s(\theta)$
and $\partial e_s(\theta)/\partial\theta = -\Phi_s$, differentiation gives
$\nabla_\theta\tr(P_0\hat W(\theta))|_{\hat\theta} = -h$.
Applying $\tr(P_0\cdot)$ to \eqref{eq:DW_param} then yields
\eqref{eq:DW_param_cost}; the remainder bound follows from
$|\tr(P_0R_k^W)| \le \|P_0\|\,\|R_k^W\|_F$.
\end{proof}

\section{From IF$^{\mathrm{m}}$ to stochastic control influence}
\label{sec:combined}

Let
\begin{equation}
    \zeta_{\hat W} := \nabla_\theta\tr(P(\theta)\hat W)\big|_{\hat\theta},
\end{equation}
computed by the fixed-covariance adjoint formula of
Lemma~\ref{lem:IF2} with $\Sigma=\hat W$.

\begin{theorem}[First-order decomposition of the LOTO plug-in cost]
\label{thm:decomp}
Under Assumptions~\ref{ass:stab}--\ref{ass:bounded_frac}, \added{and assuming
$\hat\theta_{\setminus k}\in\mathcal U$ for all $k$,} the exact LOTO plug-in
cost shift satisfies
\begin{multline}
    \Delta\hat J_k
    := \tr(P(\hat\theta_{\setminus k})\hat W_{\setminus k}) - \tr(P_0\hat W) \\
    = (\zeta_{\hat W} + h)^\top\delta\theta_k
    + \tfrac{T_k}{M_{\setminus k}}\tr\bigl(P_0(\hat W - \bar W_k)\bigr) \\
    + \mathcal R_k^{\mathrm{Ric}} + \mathcal R_k^{W} + \mathcal R_k^{\times},
    \label{eq:decomp_exact}
\end{multline}
where \added{$DP(\hat\theta)[\delta\theta_k]\in\R^{n_x\times n_x}$ is the
Fr\'echet derivative of the DARE solution $\theta\mapsto P(\theta)$ at
$\hat\theta$ in direction $\delta\theta_k$, computed by the adjoint Lyapunov
formula of Lemma~\ref{lem:IF2},}\deleted{$DP(\hat\theta)[\cdot]$ denotes the Fr\'echet derivative of $P$ at $\hat\theta$,} and
\begin{align}
    \mathcal R_k^{\mathrm{Ric}}
    &:= \tr\Bigl(\bigl[P(\hat\theta_{\setminus k}) - P_0 - D P(\hat\theta)[\delta\theta_k]\bigr]\hat W\Bigr), \\
    \mathcal R_k^{W}
    &:= \tr(P_0R_k^W), \\
    \mathcal R_k^{\times}
    &:= \tr\bigl((P(\hat\theta_{\setminus k}) - P_0)(\hat W_{\setminus k} - \hat W)\bigr).
\end{align}
Moreover,
\begin{equation}
    |\mathcal R_k^{W}|
    \le \|P_0\|\Bigl(L_\Phi^2\|\delta\theta_k\|^2
    + 4\frac{T_k}{M}L_eL_\Phi\|\delta\theta_k\|\Bigr).
    \label{eq:RW_bound}
\end{equation}
\end{theorem}

\begin{proof}
Expand
\begin{align*}
    \Delta\hat J_k
    &= \tr\bigl((P(\hat\theta_{\setminus k}) - P_0)\hat W\bigr)
    + \tr\bigl(P_0(\hat W_{\setminus k} - \hat W)\bigr) \\
    &\hspace{1.5cm}
    + \tr\bigl((P(\hat\theta_{\setminus k}) - P_0)(\hat W_{\setminus k} - \hat W)\bigr).
\end{align*}
The first term equals
$\zeta_{\hat W}^\top\delta\theta_k + \mathcal R_k^{\mathrm{Ric}}$
by first-order Taylor expansion of $\theta\mapsto\tr(P(\theta)\hat W)$ at
$\hat\theta$. The second term equals the right-hand side of
\eqref{eq:DW_param_cost}. Substituting these two identities yields
\eqref{eq:decomp_exact}. The bound \eqref{eq:RW_bound} follows from
Proposition~\ref{prop:param_shift} and
$|\mathcal R_k^W|=|\tr(P_0R_k^W)|\le \|P_0\|\,\|R_k^W\|_F$.
\end{proof}

\added{
\textbf{Reading the decomposition.}
Equation~\eqref{eq:decomp_exact} separates four mechanisms by which
trajectory~$k$ moves the plug-in cost: the dynamics channel
$\zeta_{\hat W}^\top\delta\theta_k$ (fixed-covariance), the residual channel
$h^\top\delta\theta_k$ (sensitivity of $\hat W$ to $\theta$), the exact
direct-removal term $(T_k/M_{\setminus k})\tr(P_0(\hat W-\bar W_k))$
(finite-sample, not a derivative), and three remainders---$\mathcal R_k^{\mathrm{Ric}}$
(Lyapunov-amplified), $\mathcal R_k^W$ bounded algebraically by
\eqref{eq:RW_bound}, and second-order cross term $\mathcal R_k^{\times}$. The
score $\IFstoch_k$ in Definition~\ref{def:score} keeps the first three
mechanisms and makes the hierarchy
$\IFm_k\to\IFfixed_k(\hat W)\to\IFstoch_k$ fully explicit; the only new
explicit remainder $\mathcal R_k^W$ depends directly on data and avoids a
Lyapunov series.}

\begin{definition}[Stochastic control influence built on IF$^{\mathrm{m}}$]
\label{def:score}
Using the model-side surrogate $\IFm_k$ from
Definition~\ref{def:IF0}, define
\begin{equation}
    \mathrm{IF}_k^{\mathrm{stoch}}
    := (\zeta_{\hat W} + h)^\top \IFm_k
    + \frac{T_k}{M_{\setminus k}}\tr\bigl(P_0(\hat W - \bar W_k)\bigr).
    \label{eq:IF_stoch}
\end{equation}
Equivalently,
\begin{equation}
    \mathrm{IF}_k^{\mathrm{stoch}}
    = g_k^\top H^{-1}(\zeta_{\hat W}+h)
    + \frac{T_k}{M_{\setminus k}}\tr\bigl(P_0(\hat W - \bar W_k)\bigr).
\end{equation}
Hence $\IFstoch_k$ is obtained from the same model-side
perturbation as $\IFfixed_k(\hat W)$, with the additional covariance
channels folded into the effective gradient $(\zeta_{\hat W}+h)$ and the
direct-removal correction.
\added{
The $+h$ is consistent with $\nabla_\theta\tr(P_0\hat W(\theta))|_{\hat\theta}=-h$:
combined with the sign convention
$g_k=-(1/M)\sum_{s\in\mathcal I_k}\Phi_s^\top e_s$ in $\IFm_k=H^{-1}g_k$, the chain
rule yields a $+h^\top\IFm_k$ contribution to $\Delta\hat J_k$.}
\end{definition}

\begin{corollary}[Reduction to fixed-covariance influence]
\label{cor:reduction}
If the covariance is known and held fixed at $W$, then $\hat W\equiv W$,
$h=0$, and the direct-removal term vanishes. In that case,
\begin{equation}
    \mathrm{IF}_k^{\mathrm{stoch}} = \zeta_W^\top \IFm_k
    = \IFfixed_k(W),
\end{equation}
so the stochastic score collapses exactly to the fixed-covariance control score.
\end{corollary}

\begin{proposition}[Modular approximation error]
\label{prop:error_modular}
Under the assumptions of Theorem~\ref{thm:decomp},
\begin{multline}
    \bigl|\mathrm{IF}_k^{\mathrm{stoch}} - \Delta\hat J_k\bigr|
    \le \|\zeta_{\hat W}+h\|\,\|\widetilde\delta\theta_k-\delta\theta_k\| \\
    + |\mathcal R_k^{\mathrm{Ric}}| + |\mathcal R_k^W| + |\mathcal R_k^{\times}|.
    \label{eq:error_modular}
\end{multline}
In particular, the additional error introduced by covariance estimation is
captured by $\mathcal R_k^W$ and the covariance part of $\mathcal R_k^{\times}$;
there is no Lyapunov-operator factor in the explicit bound \eqref{eq:RW_bound}.
\end{proposition}

\begin{proof}
Subtract \eqref{eq:IF_stoch} from \eqref{eq:decomp_exact} and use
$|(\zeta_{\hat W}+h)^\top(\widetilde\delta\theta_k-\delta\theta_k)|
\le \|\zeta_{\hat W}+h\|\,\|\widetilde\delta\theta_k-\delta\theta_k\|$.
\end{proof}

\begin{remark}[Interpretation]
\label{rem:interpretation}
The model-side object $\IFm_k$ answers the identification question:
``how does trajectory $k$ perturb the fitted dynamics?'' The gradient
$\zeta_{\hat W}$ then maps that perturbation through the usual
\emph{dynamics channel}, keeping $\hat W$ fixed while differentiating the
Riccati solution. The vector $h$ is the new \emph{residual channel}: it keeps
$P_0$ fixed and differentiates the residual covariance. The direct-removal term
is not a derivative at all; it is the exact finite-sample effect of removing
trajectory $k$ from the covariance average. The $\IFfixed(\hat W)$ score
accounts only for the map $\IFm_k \mapsto \zeta_{\hat W}^\top\IFm_k$.
\end{remark}

\section{Algorithm and complexity}
\label{sec:algorithm}

\begin{algorithm}[t]
\caption{IF$^{\mathrm{m}}$, IF$^{\mathrm{fixed}}$, and stochastic control influence}
\label{alg:stoch_if}
\begin{algorithmic}[1]
\REQUIRE Data $\{(\Phi_s, x_s^+)\}_{s=1}^M$, trajectory index sets
$\{\mathcal I_k\}_{k=1}^N$, cost matrices $Q,R$, regularization $\lambda$
\STATE Solve~\eqref{eq:ls_estimator} for $\hat\theta$ and factorize $H = LL^\top$
\STATE Compute residuals $e_s = x_s^+ - \Phi_s\hat\theta$
\STATE \added{$g_k \leftarrow -\frac{1}{M}\sum_{s\in\mathcal I_k}\Phi_s^\top e_s$ for all $k$}
\STATE $\hat W \leftarrow \frac{1}{M}\sum_s e_s e_s^\top$
\STATE $\bar W_k \leftarrow \frac{1}{T_k}\sum_{s\in\mathcal I_k} e_s e_s^\top$ for all $k$
\STATE Solve the DARE for $P_0$ and form $A_{\mathrm{cl}}$
\STATE Compute $\zeta_{\hat W}$ via the fixed-covariance adjoint method of
Lemma~\ref{lem:IF2} with $\Sigma = \hat W$
\STATE $h \leftarrow \frac{2}{M}\sum_s \Phi_s^\top P_0 e_s$
\STATE $v_{\IFfixed} \leftarrow L^{-\top}L^{-1}\zeta_{\hat W}$
\STATE $v_{\mathrm{stoch}} \leftarrow L^{-\top}L^{-1}(\zeta_{\hat W}+h)$
\FOR{$k=1,\dots,N$}
\STATE $\IFm_k \leftarrow L^{-\top}L^{-1}g_k$ \hfill\COMMENT{optional explicit model-side vector}
\STATE $\IFfixed_k(\hat W) \leftarrow g_k^\top v_{\IFfixed}$
\STATE $\mathrm{IF}_k^{\mathrm{stoch}} \leftarrow g_k^\top v_{\mathrm{stoch}}
+ \frac{T_k}{M_{\setminus k}}\tr\bigl(P_0(\hat W-\bar W_k)\bigr)$
\ENDFOR
\RETURN $\{\IFm_k,\IFfixed_k(\hat W),\mathrm{IF}_k^{\mathrm{stoch}}\}_{k=1}^N$
\end{algorithmic}
\end{algorithm}

\begin{corollary}[Complexity]
\label{cor:cost}
Algorithm~\ref{alg:stoch_if} computes the full hierarchy
$(\IFm,\IFfixed,\IFstoch)$ for all $N$
trajectories in
\replaced{
$O(p^3 + n_x^3 + Mpn_x + Nn_x^2 + Np)$ time.}{\begin{equation}
    O\bigl(p^3 + n_x^3 + Mpn_x + Nn_x^2 + Np\bigr)
\end{equation}
time.}
Relative to the fixed-covariance control score, the only additional
full-data pass is the computation of $h$, which costs $O(Mpn_x)$, while the
direct-removal term adds $O(n_x^2)$ work per trajectory. If one needs only
$\IFstoch$, the vectors $\IFm_k$ need not be
materialized explicitly because the score can be evaluated through the shared
solve $g_k^\top v_{\mathrm{stoch}}$.
\end{corollary}

\begin{remark}[Scalable implementation]
\label{rem:scalable}
Algorithm~\ref{alg:stoch_if} is written using an explicit Cholesky
factorization of $H$ because the present ridge-LS setting yields a structured,
symmetric positive definite Hessian. When $p$ is large, the shared solves
\replaced{
the conjugate-gradient (CG) method is used to solve}{can instead be computed by conjugate
gradients} $Hv_{\IFfixed}=\zeta_{\hat W}$ and
$Hv_{\mathrm{stoch}}=\zeta_{\hat W}+h$ \replaced{using only Hessian--vector products
\begin{equation}
  Hv = \frac{1}{M}\sum_{s=1}^M \Phi_s^\top(\Phi_s v) + \lambda v,
  \label{eq:Hv_def}
\end{equation}
without explicitly forming $H^{-1}$.}{using only Hessian-vector products
\[
Hv = \frac{1}{M}\sum_{s=1}^M \Phi_s^\top(\Phi_s v) + \lambda v,
\]
without forming $H^{-1}$ explicitly.} The per-trajectory scores then retain the
same amortized form $g_k^\top v$ plus the direct-removal correction. By
contrast, TracIn and TRAK~\cite{pruthi2020estimating,pmlr-v202-park23c} are best viewed
here as scalable attribution baselines for iterative or nonlinear training
pipelines rather than exact drop-in replacements for the linear-system solves:
they avoid inverse-Hessian computation, but they approximate a different
training-path attribution object than the first-order ridge-LS leave-one-out
surrogate used in this paper.
\end{remark}

\section{Numerical validation}
\label{sec:experiments}

We run two complementary validations: (i) control side, comparing
$\IFstoch_k$ and the baseline $\IFfixed_k(\hat W)$ against exact
leave-one-trajectory-out (LOTO) retraining; (ii) model side, validating
$\IFm_k$ through held-out prediction loss.

\subsection{Setup}

\added{
\textbf{Common.} Ridge $\lambda{=}10^{-3}$, $Q{=}I,R{=}I$,
$u_t{\sim}\mathcal N(0,I)$, $N{=}50$ trajectories ($30$ for UAV),
$T_k{\sim}\mathrm{Uniform}\{5,40\}$ ($\{20,60\}$ UAV mission), $20$ seeds
($10$ UAV). DARE solved by SciPy Schur (tol.\ $10^{-12}$); exact LOTO refits
$(\hat\theta_{\setminus k},\hat W_{\setminus k},P_{\setminus k})$ per trajectory.
\textbf{Linear plants.} Sampled planar double integrator ($\Delta t{=}0.2$):
\emph{DC motor} uses common $W{=}0.1I_4$; \emph{MSD} (stress test) uses
$\sigma_k^2{\sim}\mathrm{Uniform}(0.01,1.0)$, violating common-$W$.
\textbf{Nonlinear UAV.} Planar quadrotor RK4 at $\Delta t{=}0.05$
($m{=}0.5,J{=}10^{-3},g{=}9.81$); hover noise $0.01 I_6$, mission
$\sigma_k^2{\sim}\mathrm{Uniform}(0.005,0.1)$; identification uses the
\emph{linearized} template (out-of-model).
\textbf{Protocol.} Control side: compare $\IFstoch_k$ and the baseline
$\IFfixed_k(\hat W)$ against the exact plug-in LOTO cost change
$\Delta\hat J_k:=\tr(P_{\setminus k}\hat W_{\setminus k})-\tr(P_0\hat W)$
using Spearman $\rho_s$, Kendall $\tau$, top-5 Jaccard, and relative absolute
error $\mathrm{rae}_k:=|\IFstoch_k-\Delta\hat J_k|/(|\Delta\hat J_k|+10^{-12})$.
Model side: $10^4$ held-out transitions, $\mathrm{IF}^{\mathrm{pred}}_k
=\nabla_\theta L_{\mathrm{pred}}(\hat\theta)^\top\IFm_k$ vs exact
$\Delta L_{\mathrm{pred},k}$.}

\subsection{Results}

Table~\ref{tab:results} summarizes control ranking: even in-model,
$\IFstoch$ is much closer to exact LOTO than the $\IFfixed(\hat W)$ baseline, and
the gap widens under heterogeneous noise.

\begin{table}[t]
\centering
\caption{\replaced{
Ranking agreement with exact LOTO and per-method end-to-end runtime (single
CPU, median over 3 seeds). Mean $\pm$ std over 20 seeds (DC, MSD) / 10 seeds
(UAV); bold marks the better method.}{Ranking agreement with exact LOTO retraining. Mean $\pm$ std over 20 random seeds for the DC motor and mass--spring--damper systems, and over 10 random seeds for the two UAV tasks.}}
\label{tab:results}
\setlength{\tabcolsep}{3pt}
\resizebox{\linewidth}{!}{%
\begin{tabular}{@{}llcccc@{}}
\toprule
System & Method & Spearman $\rho_s$ & Top-5 Jaccard & LOTO ms & IF ms \\
\midrule
\multirow{2}{*}{DC motor ($p{=}24$)}
& $\IFfixed(\hat W)$ & $0.820{\pm}0.068$ & $0.430{\pm}0.234$ & \multirow{2}{*}{$39.7$} & $1.4$ \\
& $\IFstoch$ & $\mathbf{0.998{\pm}0.002}$ & $\mathbf{0.917{\pm}0.148}$ & & $2.5$ \\
\addlinespace[1pt]
\multirow{2}{*}{MSD ($p{=}24$)}
& $\IFfixed(\hat W)$ & $0.675{\pm}0.106$ & $0.404{\pm}0.148$ & \multirow{2}{*}{$44.1$} & $1.5$ \\
& $\IFstoch$ & $\mathbf{0.998{\pm}0.001}$ & $\mathbf{0.950{\pm}0.122}$ & & $2.7$ \\
\addlinespace[1pt]
\multirow{2}{*}{UAV hover ($p{=}48$)}
& $\IFfixed(\hat W)$ & $0.932{\pm}0.105$ & $0.800{\pm}0.172$ & \multirow{2}{*}{$43.6$} & $2.2$ \\
& $\IFstoch$ & $\mathbf{0.966{\pm}0.047}$ & $\mathbf{0.833{\pm}0.176}$ & & $5.8$ \\
\addlinespace[1pt]
\multirow{2}{*}{UAV mission ($p{=}48$)}
& $\IFfixed(\hat W)$ & $\mathbf{0.897{\pm}0.055}$ & $\mathbf{0.482{\pm}0.139}$ & \multirow{2}{*}{$47.7$} & $2.2$ \\
& $\IFstoch$ & $0.295{\pm}0.117$ & $0.257{\pm}0.184$ & & $5.5$ \\
\bottomrule
\end{tabular}}
\end{table}

\textbf{Control-side ranking.}
On the DC motor, $\IFstoch$ remains accurate even when the model is
correctly specified: removing a trajectory still perturbs the empirical
covariance via the exact direct-removal term. On MSD, the $\IFfixed(\hat W)$
baseline degrades because residual-covariance changes are the dominant
mechanism behind LOTO shifts. Pooling all seeds across the four settings
(Fig.~\ref{fig:control_scatter}) shows $\IFstoch_k$ tracking the exact
plug-in LOTO score in the first three regimes.

\begin{figure}[!t]
\centering
\includegraphics[width=\columnwidth]{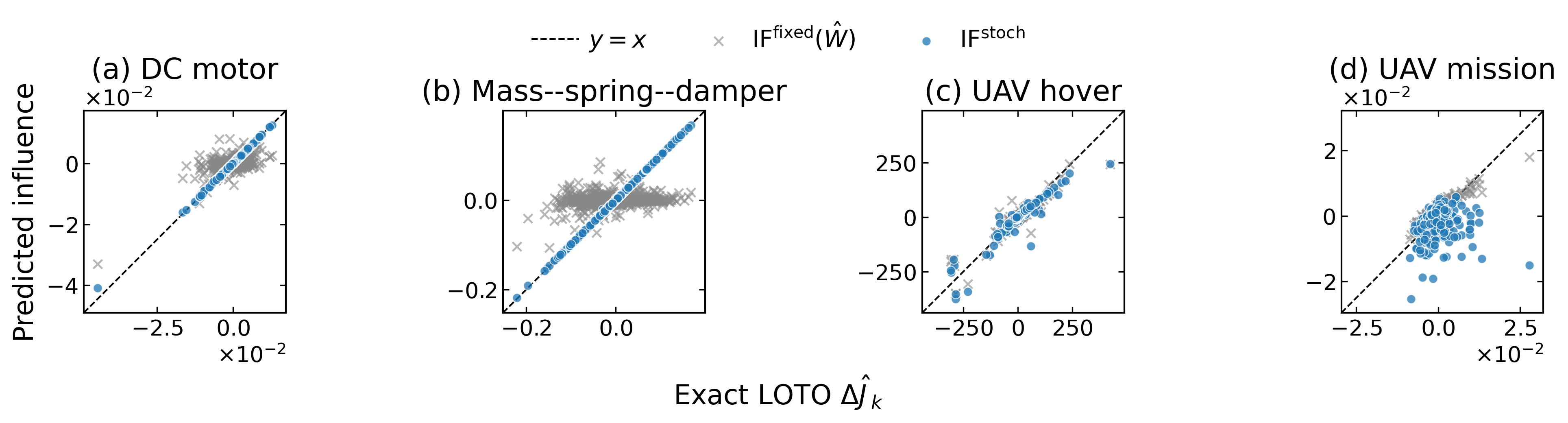}
\caption{Exact LOTO $\Delta\hat J_k$ vs.\ predicted influence,
all seeds pooled, on (a) DC motor, (b) MSD, (c) UAV hover, (d) UAV mission.
Blue $\circ$: $\IFstoch$; gray $\times$: $\IFfixed(\hat W)$; dashed $y{=}x$.
$\IFstoch$ tracks LOTO in (a)--(c); (d) is the breakdown regime where
the covariance-channel remainder $\mathcal R_k^{W}$ dominates.}
\label{fig:control_scatter}
\end{figure}

\textbf{Model-side prediction validation.}
Across $10^4$ held-out transitions, $\mathrm{IF}^{\mathrm{pred}}_k$ attains
Spearman $0.968{\pm}0.018$ ($R^2{=}0.938$, DC) and $0.962{\pm}0.022$
($R^2{=}0.897$, MSD). The model-side surrogate $\IFm_k$ therefore tracks
identification-side LOTO sensitivity even before any control gradient is
applied.

\added{
\textbf{Errors, ablations, UAV stress.} The median rae of $\IFstoch$ is
$0.04/0.03/0.18/0.81$ on DC/MSD/UAV-hover/UAV-mission. Ablations: removing $h$
retains Spearman ${>}0.95$ on MSD but degrades Top-5 Jaccard; removing the
direct-removal term hurts DC motor; removing both recovers $\IFfixed(\hat W)$.
On UAV mission, the measured decomposition gives median
$|\mathcal R^W_k|/|\Delta\hat J_k|{\approx}1.22$ versus
$0.02$ for $\mathcal R^{\mathrm{Ric}}_k$ and $0.006$ for $\mathcal R^{\times}_k$,
identifying the covariance-channel remainder as the dominant error source.}

\section{Conclusion}
\label{sec:conclusion}

\replaced{
We organized trajectory influence for stochastic LQR as a three-level
hierarchy, $\IFm_k\to\IFfixed_k\to\IFstoch_k$. On linear plants and the
linearized UAV-hover task, $\IFstoch_k$ recovers exact LOTO at $O(p)$ per
trajectory after one factorization, providing a practical substitute for
retraining in linear stochastic LQR; the UAV-mission task is a breakdown
regime in which the covariance remainder $\mathcal R_k^{W}$ dominates, and
extending the hierarchy to strongly nonlinear plants is left as future work.
\added{
\textbf{Use of Generative AI.} Generative AI was used only for
editorial and \LaTeX{} polish; all results are the authors' own.}}{We organized trajectory influence for stochastic LQR into a three-level
hierarchy. The model-side surrogate $\IFm_k$ captures how removing
trajectory $k$ perturbs the identified dynamics; the fixed-covariance control
score $\IFfixed_k$ is obtained by applying a Riccati gradient to that
perturbation; and the stochastic score augments the same model-side object with
a residual channel and an exact direct-removal covariance correction.}

\bibliographystyle{IEEEtran}
\bibliography{references}

@article{dean2020sample,
  title={On the sample complexity of the linear quadratic regulator},
  author={Dean, Sarah and Mania, Horia and Matni, Nikolai and Recht, Benjamin and Tu, Stephen},
  journal={Foundations of Computational Mathematics},
  volume={20},
  number={4},
  pages={633--679},
  year={2020},
  publisher={Springer}
}

@inproceedings{tu2019gap,
  title={The gap between model-based and model-free methods on the linear quadratic regulator: An asymptotic viewpoint},
  author={Tu, Stephen and Recht, Benjamin},
  booktitle={Conference on learning theory},
  pages={3036--3083},
  year={2019},
  organization={PMLR}
}

@article{mania2019certainty,
  title={Certainty equivalence is efficient for linear quadratic control},
  author={Mania, Horia and Tu, Stephen and Recht, Benjamin},
  journal={Advances in neural information processing systems},
  volume={32},
  year={2019}
}

@inproceedings{koh2017understanding,
  title={Understanding black-box predictions via influence functions},
  author={Koh, Pang Wei and Liang, Percy},
  booktitle={International conference on machine learning},
  pages={1885--1894},
  year={2017},
  organization={PMLR}
}

@article{pruthi2020estimating,
  title={Estimating training data influence by tracing gradient descent},
  author={Pruthi, Garima and Liu, Frederick and Kale, Satyen and Sundararajan, Mukund},
  journal={Advances in Neural Information Processing Systems},
  volume={33},
  pages={19920--19930},
  year={2020}
}

@InProceedings{pmlr-v202-park23c,
  title = 	 {{TRAK}: Attributing Model Behavior at Scale},
  author =       {Park, Sung Min and Georgiev, Kristian and Ilyas, Andrew and Leclerc, Guillaume and Madry, Aleksander},
  booktitle = 	 {Proceedings of the 40th International Conference on Machine Learning},
  pages = 	 {27074--27113},
  year = 	 {2023},
  volume = 	 {202},
  series = 	 {Proceedings of Machine Learning Research},
 }

@inproceedings{
  basu2021influence,
  title={Influence Functions in Deep Learning Are Fragile},
  author={Samyadeep Basu and Philip Pope and Soheil Feizi},
  booktitle={International Conference on Learning Representations},
  year={2021},
}

@article{bae2022if,
  title={If influence functions are the answer, then what is the question?},
  author={Bae, Juhan and Ng, Nathan and Lo, Alston and Ghassemi, Marzyeh and Grosse, Roger B},
  journal={Advances in Neural Information Processing Systems},
  volume={35},
  pages={17953--17967},
  year={2022}
}

@article{li2025influence,
  title={Influence Functions for Data Attribution in Linear System Identification and LQR Control},
  author={Li, Jiachen and Li, Shihao and Bakshi, Soovadeep and Xu, Jiamin and Chen, Dongmei},
  journal={arXiv preprint arXiv:2506.11293},
  year={2025}
}

@article{hampel1974influence,
  title={The influence curve and its role in robust estimation},
  author={Hampel, Frank R},
  journal={Journal of the american statistical association},
  volume={69},
  number={346},
  pages={383--393},
  year={1974},
  publisher={Taylor \& Francis}
}

@misc{ling1984residuals,
  title={Residuals and influence in regression},
  author={Ling, Robert F},
  year={1984},
  publisher={Taylor \& Francis}
}

@inproceedings{ghorbani2019data,
  title={Data shapley: Equitable valuation of data for machine learning},
  author={Ghorbani, Amirata and Zou, James},
  booktitle={International conference on machine learning},
  pages={2242--2251},
  year={2019},
  organization={PMLR}
}

@article{yeh2018representer,
  title={Representer point selection for explaining deep neural networks},
  author={Yeh, Chih-Kuan and Kim, Joon and Yen, Ian En-Hsu and Ravikumar, Pradeep K},
  journal={Advances in neural information processing systems},
  volume={31},
  year={2018}
}

\end{document}